\documentclass[letterpaper, 10 pt, conference]{ieeeconf} 
\IEEEoverridecommandlockouts                     
\title{Store-Forward and its implications for Proportional Scheduling}
\usepackage{amssymb,amsmath}
\usepackage{graphicx}

\newtheorem{theorem}{Theorem}
\newtheorem{corollary}{Corollary}

\newtheorem{proposition}{Proposition}

\newtheorem{remark}{Remark}

\newcommand{\mC}{{\mathcal C}}

\newcommand{\mE}{{\mathcal E}}

\newcommand{\mJ}{{\mathcal J}}

\newcommand{\mL}{{\mathcal L}}

\newcommand{\mR}{{\mathcal R}}
\newcommand{\mS}{{\mathcal S}}

\newcommand{\bR}{{\mathbb R}}
\newcommand{\bZ}{{\mathbb Z}}

\newcommand{\bE}{{\mathbb E}}
\newcommand{\bI}{{\mathbb I}}

\newcommand{\T}{\textsf{T}}

\newcommand{\vA}{{\mathbf A}}

\newcommand{\vM}{{\mathbf M}}

\newcommand{\vQ}{{\mathbf Q}}

\newcommand{\ve}{{\mathbf e}}

\newcommand{\vm}{{\mathbf m}}
\newcommand{\vn}{{\mathbf n}}

\newcommand{\vs}{{\mathbf s}}

\newcommand{\vecQ}{{\bf Q}}

\newcommand{\vecGamma}{{\boldsymbol \Gamma}}

\newcommand{\vecsigma}{{\boldsymbol \sigma}}

\author{N.S. Walton$^1$
\thanks{$^{1}$Korteweg-de Vries Institute for Mathematics,
University of Amsterdam,
Science Park 904,
1098 XH Amsterdam,
The Netherlands.
The work of Neil Walton is funded by the VENI research
programme, which is financed by the Netherlands Organisation for Scientific Research (NWO).
        {\tt\small n.s.walton@uva.nl}}%
}

\begin{document}

\maketitle
\thispagestyle{empty}
\pagestyle{empty}

\begin{abstract}
The Proportional Scheduler was recently proposed as a scheduling algorithm for
multi-hop switch networks.
For these networks, the BackPressure scheduler is the classical benchmark.
For networks with fixed routing, the Proportional Scheduler is maximum stable, myopic and, furthermore, will alleviate certain scaling issued found in BackPressure for large networks. Nonetheless,
the equilibrium and delay properties of the Proportional Scheduler has not been fully characterized. 

In this article, we postulate on the equilibrium behaviour of the Proportional Scheduler though the analysis of an analogous rule called the Store-Forward allocation. It has been shown that Store-Forward has asymptotically allocates according to the Proportional Scheduler. Further, for Store-Forward networks, numerous equilibrium quantities are explicitly calculable. For FIFO networks under Store-Forward, we calculate the policies stationary distribution and end-to-end route delay. We discuss network topologies when the stationary distribution is product-form, a phenomenon which we call \emph{product form resource pooling}. We extend this product form notion to  independent set scheduling on perfect graphs, where we show that non-neighbouring queues are statistically independent.
Finally, we analyse the large deviations behaviour of the equilibrium distribution of Store-Forward networks in order to construct Lyapunov functions for FIFO switch networks.

\end{abstract}

\section{Introduction}

Switch networks model numerous communication networks where service at a queue may inhibit the operation of other queues. Common examples include wireless ad-hoc networks \cite{taep92} and input queued switches \cite{mckeown1999achieving}. Further the paradigm is provides useful in insights in other contexts: bandwidth sharing, call centers, road traffic, data centers, manufacturing systems. At each time a scheduler makes a scheduling decision and, since load often varies, scheduling is often based on current or recent queue size information, rather than on explicit estimation of the network's long-run load. An important first order property of a ``good" scheduling policy is to be \emph{maximum stable} which states that, for the largest possible set of arrival rates, queues remain bounded (on average) and thus have an equilibrium length. 
Once this first order stability condition is established then secondary properties such as equilibrium queue length and delay can be investigated. 

The celebrated BackPressure policies of Tassiulas and Ephremedes  \cite{taep92} were the first class of scheduling policies that were proven to be maximum stable without explicit estimation of traffic load. The BackPressure policies are defined through a Lyapunov function argument, 
and provide a robust, generic approach to stabilizing a queueing network. Nonetheless, characterizing the stationary behaviour of the resulting processes can be difficult, most bounds on stationary queue length are made through the policy's associated Lyapunov function \cite{NMR05,stzopen}. Further, delay scaling along lengthy routes can be shown to be suboptimal, see \cite{BSS11}.

A scheduling policy called the Proportional Scheduler was recently proven to be maximum stable for multi-hop switch networks,  for networks with fixed routing. The Proportional Scheduler has structural advantages in comparison to BackPressure. Packets can be aggregated at first-in first-out (FIFO) queues and, as a result, the scheduler does not maintain information about the routes taken by packets in order to make a scheduling decision.  Since, for a communication network the number of routes supported by a network is, in general, orders of magnitude larger than the number of queues maintained by the network, this (in comparison to BackPressure) substantially reduces the structure required in memory to implement the policy and also leads to greater potential for decomposed implementation. Modifications of BackPressure have been developed in an effort to alleviate these issues \cite{BSS11,JJS13}. These and further consequences Proportional Scheduling are discussed in greater detail in the papers \cite{Wa14a,BDW14}. 

The principle aim of this article is to propose queue length and delay estimates for the Proportional Scheduler, through formula that can be explicitly calculated for Store-Forward networks -- a continuous time quasi-reversible queueing network which is know to asymptotically allocate resources as a Proportional Scheduler \cite{Wa09}. 
Due to the insensitivity property, the Store-Forward allocation was first analyzed by Bonald and 
Proutiere \cite{bopr04}. 
Indeed, the reason that the Proportional Scheduler mitigates some of the scaling phenomena found in 
BackPressure, is its close relationship between proportional fairness and reversible queueing 
systems (like the Store-Forward allocation). In reversible queueing networks, different packet types 
can be aggregated at each resource, and delays are explicit and depend on network load -- rather 
than network location as is found in BackPressure. 

We extend the single-class Markovian routing (or Jackson routing) considered by Bonald and Proutiere to a multi-class fixed routing (or Kelly routing). Through standard quasi-reversibility arguments we calculate the stationary distribution, queue size and delay for the Store-Forward Network.\\

\noindent \emph{Outline of Results}\\
In Theorem \ref{SF:eq}, we calculate the stationary distribution of Store-Forward networks with fixed routes where queues operate under a FIFO discipline. We then calculate the stationary delay of these Store-Forward networks, Theorem \ref{ThrmDelay}, and thus postulate the stationary delay for switch queueing networks operating under the Proportional Scheduler.

Next, we develop on the product from resource pooling work found for bandwidth sharing networks by Kang, Kelly, Lee, Williams \cite{kklw07ii,kklw07i}. For this work generalized Heavy traffic analysis and investigations into insensitivity are given by \cite{YeYao12,VZZ14}.
First, it was argued in \cite{swz11,BDW14} that independent switch components have independent stationary distributions, in heavy traffic. We extend these arguments to   investigate  Store-Forward on the interference graphs where the interference graph is perfect, and we show that the stationary queue length between non-adjacent queues are independent. See Theorem \ref{prodthrm} and Corollary \ref{coroll}. The result is somewhat surprising given  that independence is typically proven on product spaces.  

Finally, we analyse the large deviations behaviour of a stationary Store-Forward Network. We discuss how the rate function found by this large deviations analysis can be used as a Lyapunov function for FIFO networks operating under the Proportional Scheduler.

\section{Queueing Network Notation}
We let $\mJ$ index a set of queues. A schedule $\vecsigma$ is a vector in $\bZ_+^\mJ$ and we let $\mS$ be the (finite) set of schedules. We assume that the set $\mS$ is monotone in the sense that if $\vecsigma\in\mS$  and if $\tilde{\vecsigma}\in\bZ_+^\mJ$ is $\tilde{\vecsigma}\leq \vecsigma$ component-wise then $\tilde{\vecsigma}\in\bZ_+^\mJ$. 

We let $<\!\mS\! >$ be the convex combination of points in $\mS$. The set $<\!\mS\! >$ is a polytope and thus by our monotone assumption there exist a non-negative (full row rank) matrix $\vA=(A_{lj}:l\in\mL,j\in\mJ)$  such that 
\[
<\mS> = \bigg\{ \vs\in\bR_+^\mJ : \sum_{j\in\mJ} A_{lj}s_j \leq 1, l\in\mL \bigg\}.
\]
We call each facet $l\in\mL$ of the polytope $<\mS>$ a resource pool and, thus, $\mL$ is the set of resource pools. We use the notation that $j \in l$ if $A_{jl}>0$ and $j\notin l$ if $A_{jl}=0$.

A route is a finite ordered set of queues $r=(j^r_1,...,j^r_{k_r})$. We will assume that the queues visited on a route $r$ are distinct.\footnote{With a little extra notation, it is possible to visit a queue multiple times on a route.} We let the finite set $\mR$ be the set of routes. We use the notation $j\in r$ if queue $j$ is an element of route $r$.

We let $\vecQ=(Q_j : j\in\mJ)\in\bZ_+^\mJ$ give the vector of queue sizes. We must also consider the order of jobs within a queue. We let $\vecGamma_{jr}=(\Gamma_{jr}(1),...,\Gamma_{jr}(Q_j))$ counts the cumulative number of packets of route $r$ at queue $j$, that is $\Gamma_{jr}(1)=1$ if the first job in queue $j$ is from route $r$, and equal zero otherwise. Similarly, $\Gamma_{jr}(2)-\Gamma_{jr}(1)$ indicates the presence of a route $r$ job in the second position and so forth.

We let $a_r$ be the arrival rate on route $r$. Thus the load on queue $j\in\mJ$ and the load on resource pool $l\in\mL$ are respective is given by
\begin{equation}
a_j:=\sum_{r: j\in r} a_r,\qquad a_l:= \sum_{l : j\in l } A_{lj} a_j.
\end{equation}

\section{The Store-Forward Networks}
We describe a continuous time Markov chain called the Store-Forward Network, see Section \ref{SFN}. Here, queueing resources are allocated according to the Store-Forward allocation, rather than BackPressure or Proportional Fairness (cf. Section \ref{SN:PF}). 
The desirable properties that are associated with Store-Forward -- such as quasi-reversibility and product-form -- are inherited from its close relationship with quasi-reversible queueing networks. Essentially the Store-Forward allocation is defined as the throughput of a certain Kelly network. This point is discussed below in Section \ref{ClosedSF}.

Further, properties found in Store-Forward pass over to the Proportional Scheduler. This is because the asymptotic behaviour of Store-Forward is to allocate resources in the same way as the Proportional Scheduler and, as we discussed in the introduction, the Proportional Scheduler has been shown to have a number of significant structural advantages in comparison to BackPressure. These points are discussed in more detail in Sections \ref{PFSF} and \ref{SN:PF}.

We are interested in Store-Forward because of its implications for switch networks. 
Concrete implementations of Store-Forward in switch networks are possible, for instance, see \cite{swz11} where optimal queue size bounds for Input-queued switches can be proven.
However, we do not discuss direct implementation of Store-Forward in this article since we are interested in its consequences for the associated Proportional Scheduler. 
We discuss how results on Store-Forwards Networks can be phrased in the context of switch networks. Initial comments in this regard are made in Section \ref{SNSF} below, before developing this point further in subsequent sections.

\subsection{Store-Forward Networks}\label{SFN}
Given a vector of queue sizes $\vQ=(Q_j:j\in\mJ)$ the Store-Forward allocation $\vecsigma^{SF}(\vecQ)$ is defined by
\begin{equation}\label{StoreForward}
\vecsigma^{SF}_j(\vecQ) = \frac{\Phi(\vecQ-\ve_j)}{\Phi(\vecQ)},
\end{equation}
for each $j\in\mJ$, where $\ve_j\in\bZ_+^\mJ$ is the $j$th unit vector. The positive function $\Phi(\vecQ)$ has an explicit form; however, rather than give a formula for $\Phi(\vecQ)$ which at this point would be uninformative, we delay its definition until Section \ref{ClosedSF} below.

The Store-Forward allocation was first introduced by Bonald and Proutiere \cite{bopr04} and has typically been studied as a model of bandwidth sharing. For this reason, packets at a queue leave after completing service and the discipline within a queue has typically between processor-sharing.  However, in the context of switch networks, we wish packets to traverse between the queues on their route and for packets to be served in a first-in first-out (FIFO) manner. For this reason, we consider the following continuous time Markov chain, which we call a \emph{FIFO routed Store-Forward Network}.

Packets from each route $r$ arrive, respectively, as a Poisson processes of rate $a_r$ and join the first queue on route $r$, namely, $j^r_1$. Packets at each queue are assumed to have a service requirement that is independent exponentially distributed of mean 1. Given $\vecQ$ packets are served from queue $j$ at rate $\sigma^{SF}_j(\vecQ)$. The jobs within a queue are served in a first-in first-out manner and a job completing service at the front of the queue will subsequently join the back of the next queue on its route (or depart the network if the packet is at the final queue on its route).

\subsection{Kelly Networks, Store-Forward and $\Phi(\vecQ)$} \label{ClosedSF}
We now give an explicit expression for the Store-Forward allocation and the associated function $\Phi$. We do so by explaining a seemingly unrelated model of a closed queueing network. 

Consider a network of processor-sharing queues indexed by $\mL$.
Each server processes work at a unit rate.
To avoid ambiguity with the index $j$, we henceforth refer to these queues $\l\in\mL$ as \emph{pools}.
We assume that jobs from classes indexed by $\mJ$ arrive into this network at rate $a_j$, $j\in\mJ$. 
We assume that the pools are ordered, $\mL=\{ l_1,l_2,...,l_{|\mL|}\}$ , and that each arriving job visits these queues sequentially in this order. 
The load of a job at a queue depends on the class $j$ and the index $l$ of the pool. In particular, we assume that a class $j$ job at queue $l$ has a service requirement that is exponentially distributed with mean $A_{lj}^{-1}$.

It is somewhat classical queueing theory that the stationary number of jobs of each class at each pool, $\vm=(m_{lj}: j\in l, l\in\mL)$, is
\begin{equation}\label{picirc}
\pi^\circ(\vm)= \prod_{l\in\mL}\left[ (1-a_l)^{-1} { m_l \choose m_{lj} : j\in l } \prod_{j\in l} \Big( {A_{lj}a_j}\Big)^{m_{lj}}\right].
\end{equation}
We refer the read to Kelly \cite[Chapter 3]{Ke79} for a standard treatment of this result.
We condition the number of packets of each class to be $\vQ=(Q_j : j\in\mJ)$ then the (conditional) stationary distribution of packets in this network is given by
\begin{equation*}
\pi^\circ(\vm|\vQ)= \frac{1}{\Phi(\vQ)}\prod_{l\in\mL}\left[ { m_l \choose m_{lj} : j\in l } \prod_{j\in l} \Big( {A_{lj}}\Big)^{m_{lj}}\right].
\end{equation*}
where $\Phi(\vQ)$, as was briefly introduced in the previous subsection, is defined to be the 
normalizing constant achieved by summing the above terms over $\vm$ with $\sum_{l: j\in l} m_{jl} = 
Q_j$, $j\in\mJ$ (after canceling out the $a_j$ terms). Namely,
\begin{equation}\label{PhieqPi}
\Phi(\vQ)\prod_{j\in\mJ} \Big( a_j^{Q_j} \Big) = \pi^\circ \Big(\Big\{ \vm :  \sum_{l: j\in l} m_{jl} = Q_j, j\in\mJ  \Big\}\Big).
\end{equation}

We now have a formal definition for $\Phi(\vQ)$; however the significance of the constant is its relationship with the throughput with the network just described. Basically, conditional the number of customers of each class $\vQ=(Q_j:j\in\mJ)$ , the stationary throughput for class $j$ jobs in this processor-sharing network is given by the ratio
\begin{equation}\label{Whittle}
\frac{\Phi(\vQ-\ve_j)}{\Phi(\vQ)},
\end{equation}
in other words, by the Store-Forward allocation, \eqref{StoreForward}. It is well known that 
allocations of the above form, admit reversible and insensitive Markov chains, see Whittle 
\cite{Wh85}. 

For these calculations, which are only briefly described, are relatively straight-forward and we refer the reader to \cite{Ke79,Wa09} for detailed proofs.

\subsection{Proportional Fairness and Store-Forward} \label{PFSF}
Given a vector of queue sizes $\vQ=(Q_j:j\in\mJ)$, the proportional fair optimization is 
\begin{equation}
\sigma^{PF}(\vQ) \in \max_{\sigma\in <\mS>} \sum_{j\in\mJ} Q_j \log \sigma_j .
\end{equation}
The proportional fair optimization was first described by Kelly \cite{Ke97}. 
From this optimization, we can define the Proportional Scheduler for multi-hop switch networks. We do this in the next subsection. The main connection between Store-Forward and proportional fairness is that the limit of the Store-Forward allocation is a solution to Proportional fair optimization.

\begin{proposition}
\[
\sigma^{SF}(c\vQ) \xrightarrow[c\rightarrow\infty]{} \sigma^{PF}(\vQ) 
\]
\end{proposition}

This result is proven in \cite{Wa09} and earlier heuristic derivations can be found in Schweitzer \cite{Sc79}, Kelly \cite{Ke89}, Roberts and Massouli\'{e} \cite{MaRo99}. These statements generalize to further reversible allocations  Massouli\'{e} \cite{Ma07} and Walton \cite{Wa11b}.

\subsection{BackPressure and the Proportional Scheduler}\label{SN:PF}
For multi-hop networks, the BackPressure policies represent the canonical maximum stable scheduling policy. Letting $X_{jr}(t)$ be the number of route $r$ packets at queue $j$ at time $t$, the BackPressure policy is as follows
\begin{description}\item[BP1.] Calculate weights by comparing with downstream queue lengths,
\begin{equation}\label{BPwOpt}
w_j(X(t))= \max_{r: j\in r} \left\{   X_{jr}(t) - X_{j_+^rr}(t) , 0 \right\},\footnote{Here if there is no next link after $j$ on route $r$ then we set $X_{j_+^rr}=0$.}
\end{equation}
and let $r^{BP}_j(t+1)$ be the solution to this maximization. Here $j^r_+$ denotes the queue subsequent to $j$ on route $r$.
\item[BP2.] Over set of schedules $\mS$, solve the optimization
\begin{equation}\label{BPOpt}
\max_{\sigma\in\mS} \;\; \sum_{j\in\mJ} \sigma_j w_j(X(t)),
\end{equation}
and let $\sigma^{BP}(t+1)$ be an optimal solution.
\item[BP3.] If $w_j(X(t))>0$, at the next time instance schedule $\sigma_j^{BP}(t+1)$ packets from route $r^{BP}_j(t+1)$ from each queue $j\in\mJ$, else, do not schedule any packets from queue $j$.
\end{description}
Notice in the first step above, information must be exchanged along links to make a queue size 
comparison \eqref{BPwOpt}. Further, note that the policy needs to know the route of each packet at 
each queue in order to make a scheduling decision. If the number of routes is large, then this can 
be prohibitive.

However, it is argued in \cite{Wa14a,BDW14} that there are substantial advantages in implementing a proportional fair optimization in a switch network. For the multi-hop switch networks with fixed routing, the Proportional Scheduler is as follows: 
given a vector of link queue sizes $(Q_j(t): j\in\mJ)$,

\begin{description}\item[PS1.]  Over set of schedules $ <\! \mS\! >$ solve the optimization
\begin{subequations}\label{MH-PF}
\begin{align}
&\text{maximize}&& \sum_{j\in\mJ} Q_j(t) \log (s_j) \\
&\text{over}&& s\in <\! \mS\! >.
\end{align}
\end{subequations}
Let $\sigma(t+1)$ be a random variable on $\mS$ whose mean solves this optimization.
\item[PS2.] From each queue $j\in\mJ$, serve $\sigma_j(t+1)$ packets from the front of queue $j$. Each of these packets then join the back of their next downstream queue as determine by its route class.\footnote{In joining the back of queue, ties are broken arbitrarily.}
\end{description}
It is proven in the recent paper, \cite{BDW14}, that for FIFO service the proportional scheduler is 
maximum stable. When the rule PS2 is replaced with a random service discipline, maximum stability is 
proven in the paper \cite{Wa14a}.
The main structural advantage of the Proportional Scheduler over BackPressure is that it does not 
require information about the route taken by a packet in order to make a scheduling decision. This 
substantially reduces the information required for implementation and thus exhibits great potential 
for decentralized implementations. See \cite{Wa14a} and \cite{BDW14} for a more in depth discussion 
of this point.

\subsection{Switch networks and Store-Forward} \label{SNSF}
Switch networks as described in Section \ref{SN:PF} are discrete time queueing networks while the 
Store-Forward Network as defined in Section \ref{SFN} is a continuous time queue network. Thus, 
although there is an asymptotic relationship between the Proportional Scheduler and Store-Forward, 
the Store-Forward allocation is not constructed to be an implementable discrete time scheduling 
policy for switch networks and is instead used for reasons of analytic tractability. Nonetheless, 
we note that implementations of Store-Forward are possible in switch networks, we refer the read to 
\cite{swz11} for an algorithm which implements a continuous time Store-Forward network on a 
single-hop switched queueing network. Here certain optimality properties can be proven due to this 
relationship with product form queue networks.

\section{Calculations for Store-Forward Networks}
In this section we derived various stationary characteristics for FIFO routed Store-Forward 
Networks: in Section \ref{StandSt}, we give the stationary distribution of this Store-Forward 
Network; in Section \ref{Sec:Delay} we give the stationary route delay of a Store-Forward Network; 
in Section \ref{Sec:Prod} we derive product form results for the stationary Store-Forward network 
and analyse it on independent set scheduling for perfect graphs; and finally, in Section 
\ref{sec:LDP} we heuristically a large deviations principle for these networks and discuss how it 
can be used to construct a Lyapunov function used for stability analysis. For each result, we 
conjecture an analogous result for the Proportional Scheduler.

\subsection{Stationary Distribution and Stability}
\label{StandSt}
It is not hard to verify that the Store-Forward Network as described in Section \ref{SFN} is quasi-reversible and so by application of Kelly\rq{}s lemma \cite[Theorem 1.13]{Ke79} has stationary distribution as follows:

\begin{theorem}\label{SF:eq}
A FIFO routed Store-Forward Network is positive recurrent when $(a_j: j\in\mJ)\in <\mS>^\circ$ and has an equilibrium distribution of the form
\begin{equation}\label{SFEQ}
\pi(\vecQ,\vecGamma) = {\Phi(\vecQ)} \prod_{j\in\mJ}\prod_{r: j\in r} \Big({a_r}^{\Gamma_{jr}(Q_j)}\Big).
\end{equation}
\end{theorem}

The proof of this Theorem is given in the appendix. Note that an immediate consequence of the above theorem is the following:

\begin{corollary}\label{Cor1}
For a stationary FIFO routed Store-Forward Network, conditional on length of queue $j$, the route-class of a packet at the queue is independent and from route $r$ with probability 
\[
\frac{a_r}{a_j},\qquad \quad r\ni j.
\]
\end{corollary}

An important consequence of the structure of the Store-Forward allocation is that the stationary queue size $\vQ$ can be represented in terms of a number of independent random variables. 

\begin{proposition}\label{Prop1}
There are mutually independent random vectors $\vm_l=(m_{lj}: j\in l)$, $l\in\mL$, where
\[
Q_j = \sum_{l: l\in j} m_{lj}, 
\]
and each
\[
m_l: = \sum_{j\in l} m_{lj}, \quad l\in\mL,
\]
is a geometric random variable with parameter $a_l$ and, conditional on $m_l$, the random variables $(m_{lj}: j\in l)$ have a multinomial distribution with $m_l$ trials and parameters 
\[
\Big( \frac{A_{lj}a_j}{a_l} : j\in l\Big).
\]
\end{proposition}
A proof of this result is given in the Appendix.

The above proposition can be observed directly from the definition of the  Store-Forward allocation. Its consequences for proportional fair systems are first explored by Kang, Kelly, Lee and Williams \cite{kklw07ii}  in order to analyse resource pooling effects in Bandwidth sharing networks with multipath routing. It is first applied to switch systems by Shah et al \cite{swz11} to prove optimal scaling behaviour for switch networks in heavy traffic.

\begin{figure*}[tbh!]
  \centering
\includegraphics[width=0.85\textwidth]{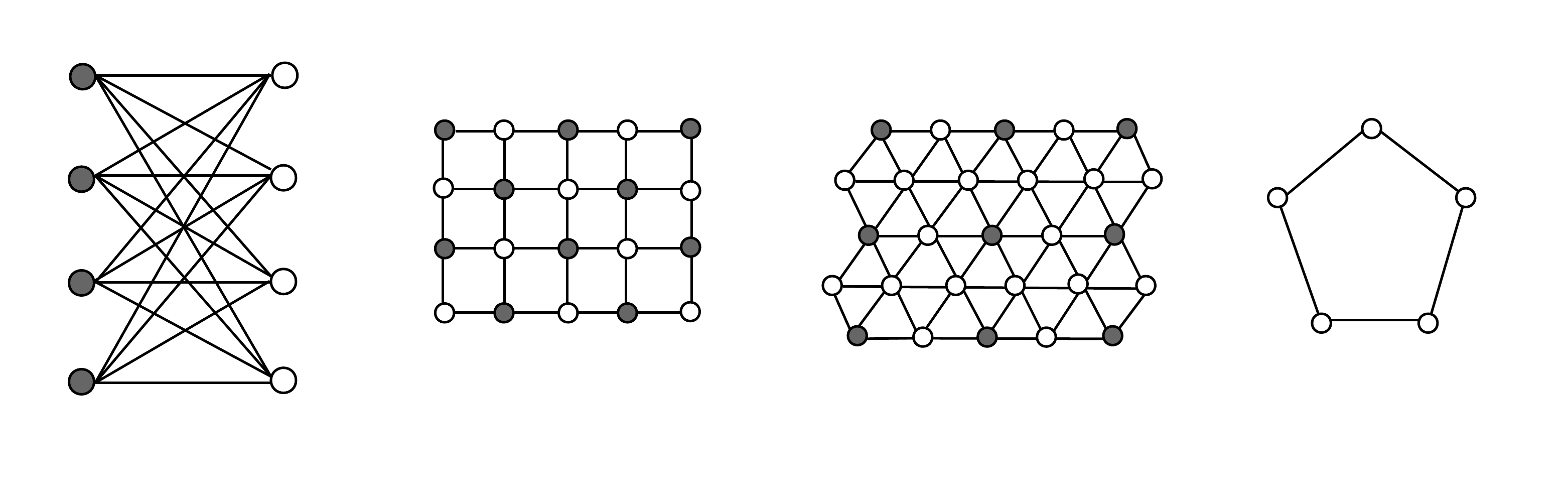}
\caption{Four examples of interference graphs: a complete bipartite graph (input-queue switch), a square grid, a triangular grid and a odd length cycle. The first three networks are perfect graphs. In each example, the nodes coloured grey are mutually independent.
\label{Graphs} }
\end{figure*}

\subsection{Delay}
\label{Sec:Delay}
With Theorem \ref{SF:eq} and Proposition \ref{Prop1}, we can analyse the delay of a packet traversing its route in a Store-Forward network as follows. 

\begin{theorem}\label{ThrmDelay}
For a stationary Store-Forward network, the delay on a route $r$ given by $D_r$ has expectation
\[
\bE\big[ D_r \big] = \sum_{j\in r} \sum_{l : j\in l} \frac{A_{lj}}{1-a_l},
\]
or, as a shorthand in matrix multiplication, letting $\bar{\vm}=((1-a_l)^{-1} : l\in\mL )$ and $\bar{n}^r=( \bI[j \in r] : j\in \mJ )$  then
\[
\bE\big[ D_r \big]  = \bar{\vm}^{\T} A \bar{\vn}^r .
\]

\end{theorem}
\begin{proof}
First by Little's Law the delay on a route, $D_r$, has expectation 
\[
\bE\big[  D_r\big]  = \frac{1}{a_r} \sum_{j\in r} \bE \big[  Q_{jr}\big] .
\]
By Corollary \ref{Cor1} above,
\[
\bE \big[ Q_{jr}\big] = \frac{a_r}{a_j} \bE \big[ Q_j\big]
\]
and by Proposition \ref{Prop1}, we know that 
\[
\bE \big[ Q_j \big]= \sum_{l : j\in l} \bE\big[  m_{lj}\big] = \sum_{l : j\in l} \frac{A_{lj} a_j}{a_l} \bE\big[  m_l \big]
\]
and
\[
\bE\big[  m_l \big]= \frac{a_l}{1-a_l}.
\]
Combining the four equalities above, we have as required,
\[
\bE \big[ D_r \big]= \sum_{j\in r} \sum_{l : j\in l} \frac{A_{lj}}{1-a_l}.
\]

\end{proof}
\begin{remark}
The result generalizes if we allow routes to visit queues more than once, in particular, if we let $\bar{\vm}=((1-a_l)^{-1}: l\in\mL)$ and we let $\bar{\vn}^r=(n^r_j : j\in\mJ)$ give the (mean) number of times a route $r$ job visits queue $j$ then the delay on that route is given by matrix multiplication
\[
\bE \big[ D_r \big] = \bar{\vm}^\T \vA \bar{\vn}^r.
\]
\end{remark}
\begin{remark}
It is well know that the delay for the BackPressure policy depends on length of the route, see \cite{BSS11,St11}
Note that unlike BackPressure, the position of a queue within a route does not effect the queue size and delay, what matters is the delay induced by the load vector $(a_l: l\in\mL)$. 

This is significant because if one makes and improvement in the service capacity in parts of the 
network then the queue sizes and, thus, delay will be reduced in that part of the network. For 
example, consider a single long route where schedules between queues do not interfere. For 
BackPressure under moderate load it is know that queue sizes grow linearly from the last queue to 
the first \cite{BSS11}, and an improvement in the service capacity of a queue -- say we double the 
service capacity of the first queue on the route -- will not reduce queue sizes below the size of 
the queue in front. However, the Store-Forward network we see from the above expression that 
doubling service capacity will halve the queue size at that queue and thus significantly reduce 
delays.
\end{remark}
\begin{remark}
Given Theorem \ref{ThrmDelay}, it is reasonable to conjecture that the same holds for the Proportional Scheduler, that is, if $D^{PS}_r$ gives the stationary route delay of the proportional scheduler, then
\begin{equation}
\bE\big[ D^{PS}_r \big] \sim \sum_{j\in r} \sum_{l : j\in l} \frac{A_{lj} }{1-a_l}
\end{equation}
as we let the load vector $(a_j : j\in\mJ)$ approach the boundary of the scheduling polytope $<\mS>$. 
\end{remark}

To understand the effect of loads on queue sizes a better understanding of the effect the scheduling set $\mS$ on the matrix $A$ is required. For certain examples, bipartite graphs and perfect graphs, the matrix $A$ has a relatively simple structure we discuss this in the next section.

\subsection{Product Form Resource Pooling}
\label{Sec:Prod}
The underlying queueing mechanism for Store-Forward is quasi-reversible. It is often found that 
quasi-reversible Markov chains exhibit product-form stationary distributions on product sets. This 
is observed to be the case for Store-Forward. For instance, in the upcoming paper \cite[Proposition 
1]{BDW14} it is proven that a Store-Forward network whose schedules are a product set, i.e.,  $\mS= 
\mS_1\! \times\! ...\! \times\! \mS_N$, the stationary queue size vectors associated with each of 
these components  are independent.

We generalize this result in the following proposition. Typically product-form results apply over 
product sets; however, an interesting consequence of this result is that the queues considered do 
not need  the scheduling set $\mS$ to be of a product type in order to have product-form stationary 
behaviour.

\begin{theorem}\label{prodthrm}
Consider a stationary Store-Forward Network on scheduling set 
\begin{equation}
<\mS> = \bigg\{ \vs\in\bR_+^\mJ : \sum_{j\in\mJ} A_{lj}s_j \leq 1, l\in\mL \bigg\}.
\end{equation}
If there are two queues $j$ and ${j'}$ such that there is no share resource pool, i.e. $\nexists$ $l\in\mL$ such that $A_{lj}>0$ and $A_{lj'}>0$, then the queues are statistically independent. 
\end{theorem}

Note when we say the queues are independent we also mean the route classes of the packets within the queue are also statistically independent. A proof of this result is given in the Appendix. Further, the result immediately extends to give the independence of sets of queues provided no pair of queues shares a common resource pool.

Now the above result may seem somewhat abstract because of the dependency on the matrix $\vA$. A good example to consider is where the matrix $\vA$ corresponds to the interference graph of a perfect graph. 
We take a graph $G=(\mJ,\mE)$ with vertices $\mJ$ and edges $\mE$. Here a queue located at each 
vertex of the graph can transmit a packet provided none of its neighbours transmit. Thus the set of 
schedules are the independent sets of this graph:
\[
\mS = \{ \sigma\in \bZ_+^\mJ : \sigma_j + \sigma_{j'} \leq 1, (j,j')\in\mE \}
\]
A graph is perfect if the neither the graph $G$ nor its complement contain an odd cycle of length $5$ or greater.  Here bipartite graphs are an important special case.
When a graph is perfect the convex hull of $\mS$ takes a explicit form 
\[
<\mS > = \bigg\{ s\in  [0,1]^\mJ : \sum_{j\in C} s_j  \leq 1, C\in\mC \bigg\}
\]
where $\mC$ gives the set of cliques of the graph $G$. This deep result is proven in \cite{chudnovsky2006strong}. So for a perfect graph, the resource pools are the cliques of the graph and to queues do not share a clique so long as they are not neighbours. Thus a direct consequence of the above theorem is the following observation
\begin{corollary}\label{coroll}
For a Store-Forward network on an interference graph which is perfect, if two queues are not 
neighbours then they are statistically independent. 
\end{corollary}

See Figure \ref{Graphs}  for some example of perfect graphs. 
The result is interesting since these graphs appear to exhibit periods of independence when mixing between disjoint modes of operation under different scheduling algorithms, e.g. for a complete bipartite graph either the left-hand side is sending or the right-side and during these periods queue sizes at nodes should be approximately independent. See \cite{Zocca2013900} and 
\cite{zocca2012mixing} for in depth discussion a CSMA scheduling algorithm on these network topologies.

Of course, it is natural to conjecture that this same independent behaviour occurs for the Proportional Scheduler under limit regimes where the network is congested, such as heavy traffic or large deviations limits.

\subsection{ Large Deviations Estimates and Lyapunov functions}
\label{sec:LDP}
Finally, we analyse the large deviations behaviour of the Proportional Scheduler. Analogous to 
Massouli\'{e} \cite{Ma07}, the  large deviations rate function found in Theorem \ref{SF:eq} provides 
a Lyapunov function that we can then use to prove stability for the Proportional 
Scheduler.\footnote{Theorem \ref{SF:eq} was proven in order to justify proceeding with the proofs in 
\cite{Wa14a} and \cite{BDW14}.} Indeed, the hope is that calculations of this type can be used to 
form Lyapunov functions of other FIFO switch scheduling policies. However, currently it is not clear 
how critical the underlying reversibility of the model is required in order for this analysis to 
work out. 

For reasons of space, we give a heuristic derivation of the rate function associated with the stationary distribution, \eqref{SFEQ}. We take a sequence of states $(\vQ^c,\vecGamma^c)$ where
\begin{equation*}
\frac{\vQ^c}{c} \xrightarrow[c\rightarrow\infty]{} \vQ 
\end{equation*}
and for a piecewise linear process $\vecGamma$
\begin{equation*}
\frac{\vecGamma^c(cq)}{c} \xrightarrow[c\rightarrow\infty]{} \vecGamma(q),
\end{equation*}
for  $q\leq \vQ_j$.
Here, convergence is point-wise (or uniform since the limit process is increasing and continuous). 
Since $\vecGamma$ is assumed to be piecewise linear, we let $k=0,...,K$ index these linear stages 
with each stage each started at queue sizes $Q_{j}(k)$, $k=0,...,K$. Let $\Delta Q_j(k) = (Q_j(k) - 
Q_j(k-1))$ and let $\vecGamma'_{jr}(k)$ index the gradient of $\vecGamma$ in each stage. By 
definition
\[
\sum_{r} \vecGamma'_{jr}(k) = 0
\]
for each $j\in\mJ$ and $k=1,...,K$. So we may interpret these gradients as probability distributions 
giving the relative density of packets along each FIFO queue.

Under $\pi$, the stationary probability that $\vecGamma^c(Q^c_j(k))/c \approx \vecGamma(Q_j(k))$, for $k=0,...,K$ and that $\vQ^c \approx c\vQ$ is approximated by

\[
P^c:= \Phi(\vQ^c) \prod_{j\in\mJ}\prod_{k=1}^K { c\Delta Q_j(k)  \choose c\Gamma'_{jr}\Delta Q_j(k) :  r\ni j }  \prod_{r\in\mR} a_r^{\Gamma'_{jr}\Delta Q_j(k)}.
\]
It can be shown that $\Phi(\vQ^c)$ approximates the solution of the proportional fair optimization in the following sense
\[
\lim_{c\rightarrow\infty} \frac{1}{c} \log \Phi(\vQ^c) = -\max_{\sigma \in <\mS>} \sum_{j\in\mJ} Q_j \log \sigma_j .
\]
A short argument for the above equality can be found in \cite{Wa11b}.  Applying a Stirling's 
approximation to the multinomial term about yields the expression
\begin{align*}
&\lim_{c\rightarrow\infty} \frac{1}{c} \log P^c \\
=& -\max_{\sigma\in <\mS>} \sum_{j\in\mJ} Q_j \log \sigma_j  \\
& + \sum_j  \sum_k \bigg[ \Delta Q_j(k) \log \Delta Q_j(k) \\
&\qquad - \sum_{r} \Gamma'_{jr} \Delta Q_j(k) \log \Big(\Gamma'_{jr} \Delta Q_j(k)\Big) \\
& \qquad + \sum_r \Gamma'_{jr} \Delta Q_j(k) \log a_r \bigg] \\
=&  -\max_{\sigma\in <\mS>} \sum_{j\in\mJ} Q_j \log \sigma_j  \\
 & - \sum_j \sum_r \sum_k \Delta Q_j(k) \Gamma'_{jr}(k) \log \Big( \frac{\Gamma'_{jr}(k)}{a_r} \Big)\\
 = &  -\max_{\sigma\in <\mS>} \sum_{j\in\mJ} Q_j \log \sigma_j   -\sum_j \sum_r \int_0^{Q_j}  \log \Big( \frac{\Gamma'_{jr}}{a_r}\Big)\;   d  \Gamma_{jr}
\end{align*}

The entropy terms derived above can be seen as a form of Sanov's Theorem. The rate function with 
this integral representation is used by Bramson as a Lyapunov function to prove stability for FIFO 
queueing networks with fixed service capacity \cite{Br96a}. Removing this integral term, a similar 
large deviations argument is used by Massouli\'e to prove stability for proportional fair networks 
with probabilistic routing \cite{Ma07}. Combining these entropy arguments, the above rate function 
provides a Lyapunov function for the Proportional Scheduler in FIFO switch networks, \cite{BDW14}. 
It is natural to conjecture that the above rate function is the large deviations rate function for a 
FIFO network operating under the Proportional Scheduler. It is currently unclear the extent to which 
the above heuristic and Lyapunov function can be developed beyond a proportional fair framework.

\bibliographystyle{IEEEtranS}
\bibliography{REFERENCES}

\appendix

\section{Additional Proofs}
We provide proofs of a number of results in the body of the text. First, we prove Theorem \ref{SF:eq}.

\begin{proof}[Proof of Theorem \ref{SF:eq}]
We verify that $\pi(\vecQ,\vecGamma)$ is a stationary distribution by confirming quasi-reversibility 
of our Markov chain. Here we must define the time-reversal of the Store-Forward network. In this 
time reversal, packets on each route $r$ arrive as a Poisson process of rate $a_r$ at the last queue 
on route $r=(j^r_1,...,j^r_{k_r})$, namely $j^r_{k_r}$. Packets have a service requirement that is 
independent exponentially distributed with mean $1$. Given the vector of queue sizes $\vecQ$ 
packets, queue are served at a rate as given by the Store-Forward allocation, 
$\vecsigma^{SF}(\vecQ)$. Queues are served in a FIFO order; however, in comparison to the 
Store-Forward network, (in forward time) jobs are served from the end of the queue and arrival are 
placed at the front of the queue. A packet of route $r$ completing service on the $k$th stage of its 
route, i.e. at queue $j^r_{k}$, then moves to the $k-1$th queue on its route, i.e. to queue 
$j^r_{k-1}$, or leaves the network if $k=1$ and thus has completed service at all queues on its 
route.

Now there are three types of transition that can occur: an arrival on route $r$; a departure on route $r$; and a transition between of a route $r$ job between queues $j$ and $j'$. We let $q((\vecQ,\vecGamma),(\vecQ',\vecGamma'))$ give the transition rates of the Store-Forward network as described in Section \ref{SFN}. We let $q^R((\vecQ,\vecGamma),(\vecQ',\vecGamma'))$ be the transition rates of the time reversal of this Store-Forward network, as described in the above paragraph.

We verify balance equations, first, for an arrival transition on route $r$ at queue $j=j^r_1$, here a transition from $(\vecQ,\vecGamma)$ to $(\vecQ',\vecGamma')$ where $\vecQ$ occurs where $Q_j'=Q_j+1$ and $\Gamma'_{jr}(Q_j+1)=\Gamma_{jr}(Q_j)+1$, all other components of $(\vecQ,\vecGamma)$ and $(\vecQ',\vecGamma')$ are equal.
\begin{align}
&\pi(\vecQ,\vecGamma)\times q((\vecQ,\vecGamma), (\vecQ',\vecGamma')) \label{Balance1}\\
=& \underbrace{\frac{\Phi(\vecQ)}{\Phi(\vecQ+\ve_j)}}_{\sigma_j^{SF}(\vecQ')} \Phi(\vecQ+\ve_j) \prod_{j\in\mJ}\prod_{r: j\in r} \Big({a_r}^{\Gamma_{jr}(Q_j)}\Big) \times a_r \\
= &  \pi(\vecQ',\vecGamma')\times q^R((\vecQ',\vecGamma'), (\vecQ,\vecGamma)) \label{Balance3}
\end{align}
The same argument taken from equation \eqref{Balance3} to \eqref{Balance1} gives the equivalent expression shows that these balance equations hold for a departure transition. Next, for a transition between queues $j$ and $j'$ on route $r$, a transition from $(\vecQ,\vecGamma)$ to $(\vecQ',\vecGamma')$ occurs with  $Q_{j'}'=Q_{j'}+1$, $Q'_j=Q_j-1$, $\Gamma'_{j'r}(Q_{j'}+1)=\Gamma_{j'r}(Q_{j'})+1$  and $\Gamma'_{jr}(Q_j)=\emptyset$, all other components  of $(\vecQ,\vecGamma)$ and $(\vecQ',\vecGamma')$ are equal.
\begin{align}
&\pi(\vecQ,\vecGamma)\times q((\vecQ,\vecGamma), (\vecQ',\vecGamma')) \label{Balance4}\\
=&\Phi(\vecQ) \prod_{j\in\mJ}\prod_{r: j\in r} \Big({a_r}^{\Gamma_{jr}(Q_j)}\Big) \times \frac{\Phi(\vecQ-\ve_j)}{\Phi(\vecQ)}\\
=& \Phi(\vecQ') \prod_{j\in\mJ}\prod_{r: j\in r} \Big({a_r}^{\Gamma'_{jr}(Q'_j)}\Big) \times \frac{\Phi(\vecQ')}{\Phi(\vecQ'+\ve_j)}\\
= & \pi(\vecQ',\vecGamma')\times q^R((\vecQ',\vecGamma'), (\vecQ,\vecGamma))  
\end{align}
Finally, for each $\vecQ$, 
\begin{align*}
\sum_{\vecQ',\vecGamma'}  q((\vecQ,\vecGamma), (\vecQ',\vecGamma')) =& \sum_r a_r  + \sum_j \sigma_j^{SF}(\vecQ) \\
=& \sum_{\vecQ',\vecGamma'}  q^R((\vecQ,\vecGamma), (\vecQ',\vecGamma')).
\end{align*}
This verifies that the conditions of Kelly's Lemma \cite[Theorem 3.1]{Ke79} hold and thus \eqref{SFEQ} gives the stationary distribution of our FIFO routed Store-Forward network. 

Finally we show that the measure $\pi$ can be normalized,
\begin{align*}
\sum_{\vecQ, \vecGamma} \pi(\vecQ,\vecGamma)
&=\sum_{\vecQ} \Phi(\vecQ) \prod_{j\in\mJ} a_j^{Q_j} \\
&= \sum_{\vecQ} \pi^\circ \Big(\Big\{ \vm :  \sum_{l: j\in l} m_{jl} = Q_j, j\in\mJ  \Big\}\Big)= 1
\end{align*}
where in the third inequality we recall \eqref{PhieqPi}.
\end{proof}

We now provide a proof of Proposition \ref{Prop1} which established a product form relation in Store-Forward Networks.

\begin{proof}[Proof of Proposition \ref{Prop1}]
It is clear that for the distribution $\pi^\circ(\vm)$, \eqref{picirc}, gives the distribution of 
the random variables $\vM_l$ as described in Proposition \ref{Prop1}. And, as given in 
\eqref{PhieqPi}, for this distribution the following identity holds
\begin{equation}\label{appr1}
\Phi(\vQ)\prod_{j\in\mJ} \Big( a_j^{Q_j} \Big) = \pi^\circ \Big(\Big\{ \vm :  \sum_{l: j\in l} m_{jl} = Q_j, j\in\mJ  \Big\}\Big).
\end{equation}
Further, the following identity holds for $\pi(\vQ)$, the stationary distribution of a Store-Forward 
network \eqref{SFEQ} ignoring the effect of the distribution of packets in the queue, $\Gamma$,
\begin{align*}
\pi(\vQ) &:= \sum_{\Gamma} \pi(\vQ,\Gamma) \\
&= \sum_{\Gamma} \Phi(\vQ) \prod_{j\in\mJ} \prod_{k=1}^{Q_j} \prod_{r\in\mR} a_r^{\Gamma_{jr}(k)-\Gamma_{jr}(k-1)}\\
& = \Phi(\vQ) \prod_{j\in\mJ} \prod_{k=1}^{Q_j} \left( \sum_{r: j\in r} a_r \right)\\
& = \pi^\circ \Big(\Big\{ \vm :  \sum_{l: j\in l} m_{jl} = Q_j, j\in\mJ  \Big\}\Big).
\end{align*}where the final equality about uses identity \eqref{appr1}. This established the relationship between $\vQ$ and independent random variables $\vm_l$, $l\in\mL$, as described.
\end{proof}

\begin{proof}[Proof of Theorem \ref{prodthrm}]
From Proposition \ref{Prop1}, we know that the stationary queue sizes of the Store-Forward Network, $\vQ$, relate to the stationary distribution of the product from queue network, $\vm$, through the equality
\[
Q_j = \sum_{l : l\in j} m_{lj}.
\]
We recall the notation that $l \in j$ when $A_{lj}>0$ and $l\notin j$ if $A_{lj}=0$. Since the random vectors $\vm_l=(m_{lj}: j\in\mJ)$ are mutually independent over $l\in\mL$. We observe that if $j$ and $j'$ do not share a common resource pool, then the summations
\[
Q_j = \sum_{l : l\in j} m_{lj},\qquad Q_{j'} = \sum_{l : l\in j'} m_{lj'}
\]
both sum over mutually exclusive indices $l$. Thus the queue sizes $Q_j$ and $Q_{j'}$ are independent. Further, since the distribution of jobs within each queue $\Gamma_{jr}$ and $\Gamma_{j'r}$ are independent when we condition on $Q_j$ and $Q_{j'}$, the queue size distribution of the queues are independent. 
\end{proof}

\end{document}